\documentclass[12pt,transaction, draftclsnofoot,onecolumn]{IEEEtran}
\usepackage{amsfonts}
\usepackage{amssymb}
\usepackage{hyperref}
\usepackage{graphicx}
\usepackage{subfigure}
\usepackage{enumerate}
\usepackage{amsmath}
\usepackage{color}
\usepackage{amsthm}
\usepackage{verbatim}
\usepackage{amsmath}
\usepackage{algorithm}
\usepackage{cite}
\usepackage{algpseudocode}

\newcommand{\bp}{\begin{proof} \small }
\newcommand{\ep}{\end{proof} \normalsize}
\newcommand{\epx}{\end{proof} \small}
\newcommand{\bpa}{\begin{proofappx} \footnotesize }
\newcommand{\epa}{\end{proofappx} \small }

\newtheorem{proposition}{Proposition}

\newtheorem{lemma}{Lemma}

\newtheorem{definition}{Definition}

\newtheorem*{theorem*}{Theorem}
\newtheorem*{proposition*}{Proposition}
\newtheorem*{corollary*}{Corollary}
\newtheorem*{lemma*}{Lemma}
\newtheorem*{assumption*}{Assumption}
\newtheorem*{definition*}{Definition}
\newtheorem*{claim*}{Claim}

\newcommand{\be}{\begin{equation}}
\newcommand{\ee}{\end{equation}}
\newcommand{\bs}{\begin{subequations}}
\newcommand{\es}{\end{subequations}}
\newcommand{\bq}{\begin{eqnarray}}
\newcommand{\eq}{\end{eqnarray}}
\newcommand{\bqn}{\begin{eqnarray*}}
\newcommand{\eqn}{\end{eqnarray*}}

\newcommand{\ba}{\left[ \begin{array}}
\newcommand{\ea}{\\ \end{array} \right]}
\newcommand{\ben}{\begin{enumerate}}
\newcommand{\een}{\end{enumerate}}

\def\real{{\mathchoice%
{\hbox{\rm\setbox1=\hbox{I}\copy1\kern-.45\wd1 R}}
{\hbox{\rm\setbox1=\hbox{I}\copy1\kern-.45\wd1 R}}
{\hbox{\scriptsize\rm\setbox1=\hbox{I}\copy1\kern-.45\wd1 R}}
{\hbox{\scriptsize\rm\setbox1=\hbox{I}\copy1\kern-.45\wd1 R}}}}

\def\Zint{{\mathchoice{\setbox1=\hbox{\sf Z}\copy1\kern-.75\wd1\box1}
{\setbox1=\hbox{\sf Z}\copy1\kern-.75\wd1\box1}
{\setbox1=\hbox{\scriptsize\sf Z}\copy1\kern-.75\wd1\box1}
{\setbox1=\hbox{\scriptsize\sf Z}\copy1\kern-.75\wd1\box1}}}
\newcommand{\complex}{ \hbox{\rm C\kern-0.45em\rule[.07em]{.02em}{.58em}%
\kern 0.43em}}

\begin{document}
%
\title{Bandwidth Allocation for Multiple Federated Learning Services in Wireless Edge Networks}
%
%
%

\author{Jie~Xu, Heqiang~Wang, Lixing~Chen
\thanks{J. Xu, H. Wang and L. Chen are with the Department of Electrical and Computer Engineering, University of Miami, Coral Gables, FL, USA. }
}

\maketitle

\begin{abstract}
This paper studies a federated learning (FL) system, where \textit{multiple} FL services co-exist in a wireless network and share common wireless resources. It fills the void of wireless resource allocation for multiple simultaneous FL services in the existing literature. Our method designs a two-level resource allocation framework comprising \emph{intra-service} resource allocation and \emph{inter-service} resource allocation. The intra-service resource allocation problem aims to minimize the length of FL rounds by optimizing the bandwidth allocation among the clients of each FL service. Based on this, an inter-service resource allocation problem is further considered, which distributes bandwidth resources among multiple simultaneous FL services. We consider both cooperative and selfish providers of the FL services. For cooperative FL service providers, we design a distributed bandwidth allocation algorithm to optimize the overall performance of multiple FL services, meanwhile cater to the fairness among FL services and the privacy of clients. For selfish FL service providers, a new auction scheme is designed with the FL service owners as the bidders and the network provider as the auctioneer. The designed auction scheme strikes a balance between the overall FL performance and fairness. Our simulation results show that the proposed algorithms outperform other benchmarks under various network conditions.

\end{abstract}


%
\IEEEpeerreviewmaketitle

\section{Introduction}
Today's mobile devices are generating an unprecedented amount of data every day. Leveraging the recent success of machine learning (ML) and artificial intelligence (AI), this rich data has the potential to power a wide range of new functionalities and services, such as learning the activities of smart phone users, predicting health events from wearable devices or adapting to pedestrian behavior in autonomous vehicles. With the help of multi-access edge computing (MEC) servers, ML models can be quickly trained/updated using this data to adapt to the changing environment without moving the data to the remote cloud data center, which is envisioned in intelligent next-generation communication systems \cite{park2019wireless}. Furthermore, due to the growing storage and computational power of mobile devices as well as privacy concerns associated with uploading personal data, it is increasingly attractive to store and process data directly on mobile devices. Federate learning (FL) \cite{konevcny2016federated} is thus proposed as a new distributed ML framework, where mobile devices collaboratively train a shared ML model with the coordination of an edge server while keeping all the training data on device, thereby decoupling the ability to do ML from the need to upload/store the data to/in a public entity. 

A typical FL service involves a number of mobile devices (a.k.a., participating clients) and an edge server (a.k.a., a parameter server) to train a ML model, which lasts for a number of learning rounds. In each round, the clients download the current ML model from the server, improve it by learning from their local data, and then upload the individual model updates to the server; the server then aggregates the local updates to improve the shared model. For example, the seminal work \cite{mcmahan2017communication} proposed the FedAvg algorithm in which the global model is obtained by averaging the parameters of local models. Although other FL algorithms differ in the specifics, the majority of them follow the same procedure. Because the clients work in the same wireless network to download and upload models, how to allocate the limited wireless bandwidth among the participating clients has a crucial impact on the resulting FL speed and efficiency. Therefore, resource allocation for wireless FL systems is attracting much attention recently in the wireless communications community \cite{chen2020joint, tran2019federated}. Compared to resource allocation in traditional throughput-maximizing wireless networks, the resource allocation objective and outcome become considerably different for wireless FL due to its unique requirements and characteristics. 

\begin{figure}[ht]
	\centering
	\includegraphics[width=0.55\textwidth]{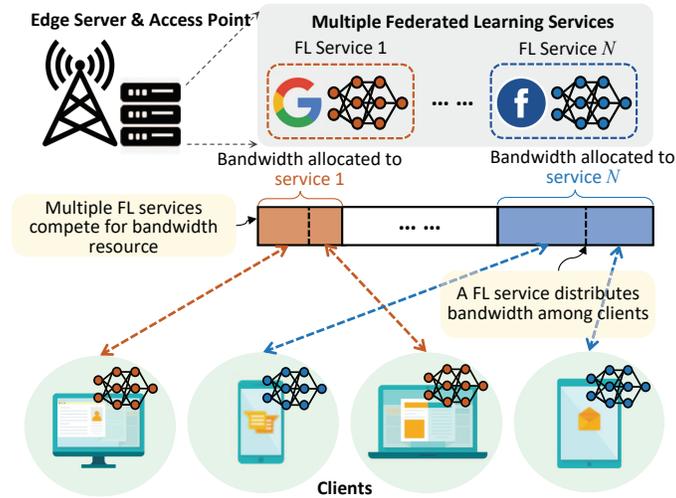}
	\caption{\small System Overview.}
		\vspace{-0.1in}
	\label{fig:system}
\end{figure}

Although existing works have made meaningful progress towards efficient resource allocation for wireless FL, they share the common limitation that only a \textit{single} FL service was considered. As ML-powered applications grow and become more diverse, it is anticipated that the wireless network will host \textit{multiple} co-existing FL services, the set of which may also dynamically change over time. See Figure \ref{fig:system} for an illustration of the multi-FL service scenario. The presence of multiple FL services makes resource allocation for wireless FL much more challenging. \textbf{First}, the achievable FL performance depends on not only \textit{intra-service} resource allocation among the participating clients within each FL service but also \textit{inter-service} resource allocation among different FL services, and these two levels of allocation decisions are also strongly coupled. \textbf{Second}, the FL service providers may adopt different FL algorithms and choose different configurations (e.g., number of participating clients, number of epochs of local training, etc.), yet this information is not always available to the wireless network operator due to privacy concerns when making resource allocation decisions.  \textbf{Third}, because FL service providers have their individual goals, they may have incentives to untruthfully reveal their valuation of the wireless bandwidth if by doing so they gain advantages in the inter-service bandwidth allocation. Without the correct information, there is no guarantee on the overall system performance. \textbf{Finally}, as in any multi-user system, resource allocation should strike a good balance between efficiency and fairness -- every FL service provider should obtain a reasonable share of the wireless resource to train their ML models using FL. 

In this paper, we make an initial effort to study wireless FL with multiple co-existing FL services, which share the same bandwidth to train their respective ML models. Our focus is on the efficient bandwidth allocation among different FL services as well as among the participating clients within each FL service, thereby understanding the interplay between these two levels of allocation decisions. Our main contributions are summarized as follows:
\begin{itemize}
    \item We formalize a two-level bandwidth allocation problem for multiple FL services co-existing in the wireless network, which may start and complete at different time depending on their own demand and FL requirements. The model is general enough for any FL algorithm that involves downloading, local learning, uploading and global aggregation in each learning round, and hence has wide applicability in real-world systems. In addition, we explicitly take fairness into consideration when optimizing bandwidth allocation to ensure that no FL service is starved of bandwidth. 
    \item We consider two use cases depending on the nature/goals of the FL service providers. In the first case, FL service providers are fully cooperative to maximize the overall system performance. For this, we design a distributed optimization algorithm based on dual decomposition to solve the two-level bandwidth allocation problem. The algorithm keeps all FL-related information at the individual FL service provider side without sharing it with the network operator, thereby reducing the communication overhead and enhancing privacy protection.
    \item We further consider a second case where FL service providers are selfishly maximizing their own performance. To address the selfishness issue, we design a multi-bid auction mechanism, which is able to elicit the FL service providers' truthful valuation of bandwidth based on their submitted bids. With a fairness-adjusted \textit{ex post} charge, the proposed auction mechanism is able to make a tunable trade-off between efficiency and fairness. 
\end{itemize}
The rest of this paper is organized as follows. Section II discusses related works. Section III builds the system model. Section IV formulates the problem for the cooperative case and develops a distributed bandwidth allocation algorithm. Section V studies the selfish service providers case and develops a multi-auction mechanism. Section VI performs simulations. Concluding remarks are made in Section VII. 

\section{Related Work}
A lot of research has been devoted to tackling various challenges of FL, including but not limited to developing new optimization and model aggregation methods \cite{karimireddy2019scaffold, li2018federated, haddadpour2019convergence}, handling non-i.i.d. and unbalanced datasets \cite{li2019convergence, zhao2018federated, sattler2019robust}, dealing with the straggler problem \cite{smith2017federated}, preserving model and data privacy \cite{fung2018mitigating, kim2019blockchained}, and ensuring fairness \cite{mohri2019agnostic, li2019fair}. A comprehensive review of these challenges can be found in \cite{li2020federated, lim2020federated, yang2019federated}. In particular, the communication aspect of FL has been recognized as a primary bottleneck due to the tension between uploading a large amount of model data for aggregation and the limited network resource to support this transmission, especially in a wireless environment. In this regard, early research on communication-efficient FL largely focuses on reducing the amount of transmitted data while assuming that the underlying communication channel has been established, e.g., updating clients with significant training improvement \cite{chen2018lag}, compressing the gradient vectors via quantization \cite{lin2017deep}, or accelerating training using sparse or structured updates \cite{aji2017sparse}. More recent research starts to address this problem from a more communication system perspective, e.g., using a hierarchical FL network architecture \cite{liu2020client} that allows partial model aggregation, and leveraging the wireless transmission property to perform analog model aggregation over the air \cite{yang2020federated}. 

As wireless networks are envisioned as a main deployment scenario of FL, wireless resource allocation for FL is another active research topic. Many existing works \cite{wang2019adaptive, zhan2020experience, tran2019federated} study the trade-off between local model update and global model aggregation. Client selection is essential to enable FL at scale and address the straggler problem. Different types of joint bandwidth allocation and client scheduling policies \cite{ xu2020client, zeng2020energy, shi2020device, nishio2019client, chen2020convergence} have been proposed to either minimize the training loss or the training time. In all these works, resource allocation is carried out among clients of a \textit{single} FL service, while assuming that the FL service itself has already received dedicated resource. In stark contrast, our paper studies a network consisting of multiple co-existing FL services and performs resource allocation at both the FL service level and the client level. We notice that a related problem where multiple FL services are being trained at the same time is also considered in a recent work \cite{nguyen2020toward}. In that paper, different FL services run on the same set of clients and a joint computation and communication resource scheduling problem is studied. In our paper, different FL services have their separate client sets which may experience very different channel qualities and hence we focus only on the bandwidth allocation problem. Moreover, while \cite{nguyen2020toward} assumed that all clients are obedient, we study the possible selfish nature of FL service providers and highlight bandwidth allocation fairness. 

Considering each FL service as a ``user'', our problem is a special type of resource allocation problems for multi-user wireless networks. While many concepts and techniques adopted in this paper, including proportional fairness \cite{massoulie1999bandwidth}, dual decomposition \cite{palomar2006tutorial} and multi-bid auction \cite{maille2004multibid}, have seen applications in other multi-user wireless resource allocation domains, applying them in multi-service FL requires special treatment as two levels of resource allocation are involved in our problem. 
In particular, there is no closed-form expression of how the performance (i.e., learning speed) of a FL service depends on the resource allocation among its clients. Therefore, understanding the inter-dependency of intra-service and inter-service bandwidth allocation is essential. Furthermore, we put an emphasis on the resource fairness among different FL services by designing a new fairness-adjusted multi-bid auction mechanism in the selfish FL service provider case, thereby achieving a tunable tradeoff between efficiency and fairness. We point out that there are some existing works \cite{feng2019joint, sarikaya2019motivating, kang2019incentive, le2020auction} on designing incentive mechanisms for client participation of a single FL service. These works are very different from our paper in terms of both the problem and the approaches, and do not consider fairness when designing the mechanism.

\section{System Model}
We consider a wireless network where machine learning models are trained using Federated Learning (FL). The wireless network has a total bandwidth $B$, and the network operator has to allocate this bandwidth among concurrent FL services when needed to enable their individual training. Because new FL services may start and old FL services may finish over time, bandwidth allocation has to be periodically performed to adapt to the current active FL services. Therefore, we divide time into periods and let the length of a period be $T$. At the beginning of each period $i$, a set $\mathcal{N}_i$ of FL services are active and require wireless bandwidth to carry out their training. These services are either newly initiated services in period $i$ or continuing services from the previous period. A FL service finishes and hence exits the wireless network when a certain termination criteria is satisfied (e.g., the training loss is below a threshold, the testing accuracy is above a threshold, or other convergence criterion), which usually varies across FL services and are pre-specified by the corresponding service provider. Therefore, a FL service may span multiple periods. The wall clock time (i.e. the number of periods) that a FL service takes to finish depends on the difficulty and other inherent characteristics of the service itself as well as how much wireless resource is allocated to this service in each period for which it stays and how this bandwidth is further allocated among its participating clients. In what follows, we first formulate the client-level (i.e., intra-service) bandwidth allocation problem and then describe the service-level (i.e., inter-service) bandwidth allocation problem.

\subsection{Intra-Service Bandwidth Allocation}
To understand how bandwidth allocation affects FL performance, let us consider a single representative FL service $n$ in one period (period index $i$ is dropped for conciseness). Suppose that this service is allocated with a bandwidth $b_n$ in this period, which is further allocated among its participating clients, the set of which is denoted by $\mathcal{K}_n$. For each client $k \in \mathcal{K}_n$, let $\phi_k$ be its computing speed, and $g^\text{ul}_k$ and $g^\text{dl}_k$ be the uplink and downlink wireless channel gains to the parameter server of service $n$, respectively, which are assumed to be invariant within a period. We consider a synchronized FL model for each FL service, where a number of FL rounds take place in a period. Nonetheless, different FL services do not have to be synchronized -- they learn at their own pace. See Figure \ref{fig:timeline} for an illustration.

\begin{figure}[ht]
	\centering
	\includegraphics[width=0.8\textwidth]{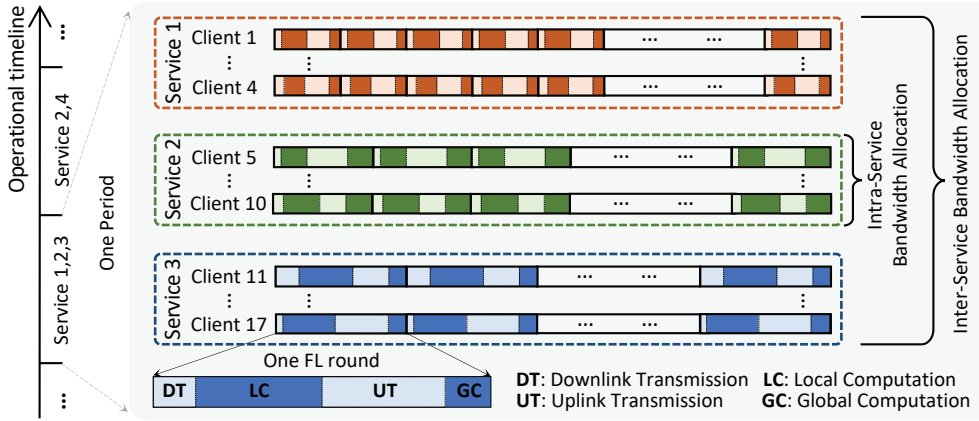}
	\caption{\small Bandwidth Allocation among Multiple FL Services.}
		\vspace{-0.1in}
	\label{fig:timeline}
\end{figure}

A FL round consists of four stages: download transmission, local computation, upload transmission and global computation:
\begin{itemize}
\item \textit{Download Transmission (DT)}. Each FL round starts with a DT stage in which each client $k$ downloads the current global model from its parameter server residing on the base station. Suppose client $k$ is allocated with bandwidth $b_{n,k}$, then its DT rate is $b_{n,k}\log_2(1+P_n g^\text{dl}_k/N_0)$ following Shannon's equation, where $P_n$ is the transmission power of parameter server $n$ and $N_0$ is the noise power. For notational convenience, we denote $\log_2(1+P_n g^\text{dl}_k/N_0)\triangleq r^\text{DT}_k$ as the DT base rate of client $k$. Let $s^\text{DT}_n$ be the download data size (e.g., the size of the global model), then the DT latency is $t^\text{DT}_{n,k} = s^\text{DT}_n/(b_{n,k} r^\text{DT}_k)$. 
\item \textit{Local Computation (LC)}. With the current global model, each client $k$ then updates its local model using its local dataset. Depending on the ML model complexity, the local dataset size and the number of episodes in local training, the per-round local computation workload is denoted by $w^\text{LC}_{n,k}$. Therefore, the LC latency of client $k$ is $t^\text{LC}_{n,k} = w^\text{LC}_{n,k}/\phi_k$. 
\item \textit{Upload Transmission (UT)}. Once local update is finished, client $k$ transmits the result to the parameter server $n$. Given the bandwidth $b_{n,k}$, its UT rate is $b_{n,k}\log_2(1+P_k g^{ul}_k/N_0)$, where $P_k$ is the transmission power of client $k$ and $N_0$ is the noise power. Again, for notational convenience, we denote $\log_2(1+P_k g^{ul}_k/N_0) \triangleq r^\text{UT}_k$ as the UT base rate of client $k$. Let $s^\text{UT}_n$ be the data size that has to be transmitted to the parameter server, then the UT latency of client $k$ is $t^\text{UT}_{n,k}=s^\text{UT}_n/(b_{n,k} r^\text{UT}_k)$. 
\item \textit{Global Computation (GC)}. Finally, once the local updates of all clients are received by parameter server $n$, the global model is updated. Let $w^\text{GC}_n$ be the global model update workload and $\phi_n$ be the computing speed of parameter server $n$, then the GC latency is $t^\text{GC}_n = w^\text{GC}_n/\phi_n$. 
\end{itemize}

Note that our framework is applicable to a vast set of FL algorithms (e.g., FedAvg, FedSGD) that can be chosen for service $n$. For instance, the downloaded/uploaded data may be the model itself, the compressed version of the model, or the model gradient information.  For the purpose of bandwidth allocation, it is sufficient to describe the FL service as a tuple $\langle s^\text{DT}_n, \{w^\text{LC}_{n,k}\}_{k\in\mathcal{K}_n}, s^\text{UT}_n, w^\text{GC}_n\rangle$. 

In synchronized FL, the parameter server updates the global model until it has received the local updates from all participating clients. Hence, the length of a FL round of service $n$ is determined by the total latency of the \textit{slowest} client, i.e. $t_n = \max_{k\in\mathcal{K}_n} (t^\text{DT}_{n,k} + t^\text{LC}_{n,k} + t^\text{UT}_{n,k} + t^\text{GC}_n)$. To minimize the FL round length $t_n$ of service $n$ so that more FL rounds can be executed in a period, one has to optimally allocate bandwidth $b_n$ among the clients of service $n$. Given $b_n$, the \textit{intra-service bandwidth allocation} problem can be formulated as
\begin{align}
    \min_{b_{n,1},..., b_{n,K}}~t_n(\{b_{n,k}\}_{k\in\mathcal{K}_n})~~~~~~
    \text{subject to}~ \sum_{k\in\mathcal{K}_n} b_{n,k} = b_n
    \label{clientBA}
\end{align}
Let $t^*_n(b_n)$ denote the optimal solution to Eqn. \eqref{clientBA}. Then the optimal FL frequency of service $n$ is $f^*_n(b_n) = 1/t^*_n(b_n)$, which is used to represent the FL speed of service $n$. Note that this means $T\cdot f^*_n(b_n)$ FL rounds can be performed in one period.

\subsection{Inter-Service Bandwidth Allocation}
In a period, multiple active FL services may be active and require wireless bandwidth to carry out learning. Since they share a total bandwidth $B$, how this bandwidth is allocated among different services will determine their achievable learning frequencies $f^*_n(b_n)$, thus the convergence speed in terms of the wall clock time. In this paper, we consider two scenarios depending on the goals of the FL service providers and how inter-service bandwidth allocation is implemented. In the first scenario, all FL service providers are \textit{cooperative}, and their goal is to maximize the FL performance of the overall system. Therefore, it is equivalent to the network operator solving a system-wide optimization problem. In the second scenario, the FL service providers are \textit{selfish} who care about only their own FL performance. As these service providers are competing for the limited bandwidth resource, addressing their incentive issues is crucial. In this paper, we design a fairness-adjusted multi-bid auction mechanism for the inter-service bandwidth allocation in this case. In the following two sections, we discuss these two scenarios separately. 

\section{Cooperative service providers}
In the cooperative service providers scenario, the network operator directly decides the bandwidth allocation to maximize the overall system performance. As in any multi-user network, bandwidth allocation for multi-service FL has to address both efficiency and fairness -- every active FL service should get a reasonable share of the bandwidth. Thus, we adopt the notion of \textit{proportional fairness} \cite{massoulie1999bandwidth}, a metric widely used in multi-user resource allocation, and aim to solve the following optimization problem:
\begin{align}
    &\max_{b_1, ..., b_N}~~\sum_{n=1}^N \log (1 + f^*_n(b_n))\nonumber\\
    \text{subject to}~~&\sum_{n=1}^N b_n = B~~\text{and}~~f^*_n(b_n)~\text{solves}~\eqref{clientBA}, \forall n
    \label{serviceBA}
\end{align}
where we drop the period index $i$ and let $N$ be the number of active FL services in the period for conciseness. The objective function adds a ``1'' inside the logarithmic to ensure that the function value is always non-negative. This change has very little impact on the final allocation since the frequency is often much larger than 1 in a period. Note also that the above inter-service bandwidth allocation problem Eqn.~\eqref{serviceBA} implicitly incorporates the intra-service problem as $f^*_n(b_n)$ is the solution to Eqn.~\eqref{clientBA}.

\subsection{Optimal Solution to the Intra-Service Problem}
We first investigate the optimal solution to the intra-service bandwidth allocation problem and see how it can be used to solve the inter-service problem. According to our system model and Eqn. \eqref{clientBA}, the intra-service bandwidth allocation is equivalent to
\begin{align}
    \min_{b_{n,1}, ..., b_{n,K}}~~&t_n\\
    \text{subject to}~~&t^\text{C}_{n,k} + \alpha_{n,k}/b_{n,k} \leq t_n\\
    &\sum_{k}b_{n,k} = b_n
\end{align}
where we let $t^\text{C}_{n,k} \triangleq t^\text{LC}_{n,k} + t^\text{GC}_n$ and $\alpha_{n,k} \triangleq s^\text{DT}_n/r^\text{DT}_k + s^\text{UT}_n/r^\text{UT}_k$ for notational convenience. Clearly, the optimal solution $t^*$ must satisfy
\begin{align}
    t^\text{C}_{n,k} + \alpha_{n,k}/b_{n,k} = t^*_n, \forall k 
\end{align}
Therefore, the optimal $t^*_n$ solves the following equality,
\begin{align}
    \sum_{k} \frac{\alpha_{n,k}}{t^*_n - t^\text{C}_{n,k}} = b_n
    \label{optT}
\end{align}
Although we do not have a closed-form solution of $t^*_n(b_n)$, a bi-section algorithm can be constructed to easily solve the above problem to obtain the optimal $t^*_n(b_n)$ and consequently the optimal frequency $f^*_n(b_n) = 1/t^*_n(b_n)$ as a function of $b_n$. Furthermore, the property of $f^*_n(b_n)$ can be characterized in the following lemma. 
\begin{lemma}
$f^*_n(b_n)$ is a differentiable, increasing and concave function for $b_n > 0$. 
\end{lemma}
\begin{proof}
Let us consider the inverse function $b_n(f_n)$ defined by Eqn. \eqref{optT}. It is easy to see that for $f_n \in [0, 1/\max_{k} t^\text{C}_{n,k})$, $b_n(F_n)$ is a monotonically increasing function in $f_n$ with $b_n(0) = 0$ and $b_n(f_n) \to \infty$ as $f_n \to 1/\max_{k} t^{cp}_{n,k}$. Therefore, for $b_n \geq 0$, $f_n(b_n)$ is also monotonically increasing. The first-order derivative of $b_n(f_n)$ is
\begin{align}
    b'_n = \frac{d b_n}{d f_n} = \frac{d b_n}{d t_n}\frac{d t_n}{d f_n} = \sum_k \frac{\alpha_{n,k}}{(1 - t^\text{C}_{n,k} f_n)^2} > 0, \forall f_n \in [0, 1/\max_k t^\text{C}_{n,k})
\end{align}
Therefore, $f_n(b_n)$ is differentiable for $b_n \geq 0$ and
\begin{align}
    f'_n = \frac{d f_n}{d b_n} = \left(\sum_k \frac{\alpha_{n,k}}{(1 - t^\text{C}_{n,k} f_n)^2}\right)^{-1} > 0, \forall b_n \geq 0
\end{align}
The second-order derivative $f''_n$ can also be computed as follows:
\begin{align}
    f''_n = -\left(\sum_k \frac{\alpha_{n,k}}{(1 - t^{C}_{n,k} f_n)^2}\right)^{-2}\left(\sum_k \frac{\alpha_{n,k} t^\text{C}_{n,k}}{(1-t^\text{C}_{n,k}f_n)^3}\right) < 0, \forall B \geq 0
\end{align}
This proves that $f_n(b_n)$ is a concave function for $b_n \geq 0$.
\end{proof}
With Lemma 1, it is straightforward to see that the inter-service bandwidth allocation problem \eqref{serviceBA} is a convex optimization problem.
\begin{proposition}
The inter-service bandwidth allocation problem \eqref{serviceBA} is an equality-constrained convex optimization problem.
\end{proposition}
\begin{proof}
Because $f^*$ is concave, $\log$ is concave and increasing, the composition $\log (1+f^*)$ is also a concave function. Then it is straightforward to see that the problem is a concave maximization problem with an equality constraint. 
\end{proof}

\subsection{Distributed Algorithm for Inter-Service Bandwidth Allocation}
We now proceed with solving the inter-service bandwidth allocation problem. While various centralized algorithms, such the Newton's method, can efficiently solve the inter-service problem Eqn. \eqref{serviceBA} given the fact that it is a convex optimization problem, we prefer a distributed algorithm where individual FL service providers do not share their FL algorithm details and client-level information with each other or the network operator. This way reduces the communication overhead and preserves privacy of the client devices of individual FL service providers. Our algorithm is developed based on dual decomposition \cite{palomar2006tutorial} as follows. 

We first relax the total bandwidth constraint $\sum_n b_n = B$ to be $\sum_n b_n \leq B$, and then form the Lagrangian by relaxing the coupling constraint:
\begin{align}
    L(b_1, ..., b_N, \lambda) = \sum_{n} \log (1 + f^*_n(b_n)) - \lambda\left(\sum_n b_n - B\right) = \sum_n L_n(b_n, \lambda) + \lambda B
\end{align}
where $\lambda$ is the Lagrange multipier associated with the total bandwidth constraint, and $L_n(b_n, \lambda) = \log (1 + f^*_n(b_n)) - \lambda b_n$ is the Lagrangian to be maximized by service provider $n$. Such dual decomposition results in each service provider $n$ solving, for a given $\lambda$, the following problem
\begin{align}
    b^*_n(\lambda) = \arg\max_{b_n\geq 0} L_n(b_n, \lambda) = \arg\max_{b_n\geq 0}\left(\log (1 + f^*_n(b_n)) - \lambda b_n\right)
    \label{optBLambda}
\end{align}
where the solution is unique due to the strict concavity of $f^*_n$ according to Lemma 1. Specifically, to solve this maximization problem, we only need to solve its first-order condition,
\begin{align}
    {f_n^*}'(b_n)/(1 + f^*_n(b_n)) = \lambda
\end{align}
which can be converted to solve $f^*$ using
\begin{align}
    (1 + f^*_n) \sum_{k\in\mathcal{K}_n} \frac{ \alpha_{n,k}}{(1 - t^\text{C}_{n,k} f^*_n)^2} = \lambda^{-1}
    \label{optFreqLambda}
\end{align}
Clearly, the left-hand side is an increasing function of $f^*_n$ for $f^*_n \in [0, 1/\max_k t^\text{C}_{n,k})$ and thus, a simple bi-section algorithm can be devised to solve Eqn. \eqref{optFreqLambda} to obtain $f^*_n(\lambda)$. Then plugging $f^*_n(\lambda)$ (hence $t^*_n(\lambda)$) into Eqn. \eqref{optT} yields the optimal $b^*_n(\lambda)$. 

Let $g_n(\lambda) = \max_{b_n\geq 0} L_n(b_n, \lambda) = L_n(b^*_n(\lambda), \lambda)$ be the local dual function for service provider $n$. Then the master dual problem is
\begin{align}
    \min_\lambda g(\lambda) = \sum_n g_n(\lambda) + \lambda B~~~~\text{subject to}~~\lambda \geq 0
\end{align}
Since $b^*_n(\lambda)$ is unique, it follows that the dual function $g_n(\lambda)$ is differentiable and the following gradient method can be used to iteratively update $\lambda$:
\begin{align}
    \lambda(j+1) = \left[\lambda(j) - \gamma\left(B - \sum_n b^*_n(\lambda(j))\right)\right]^+
    \label{updateLambda}
\end{align}
where $j$ is the iteration index, $\gamma > 0$ is a sufficiently small positive step-size, and $[\cdot]^+$ denotes the projection onto the non-negative orthant. The dual variable $\lambda(j)$ will converge to the dual optimum $\lambda^*$ as $j \to \infty$. Since the duality gap for the inter-service problem Eqn. \eqref{serviceBA} is zero and the solution to Eqn. \eqref{optBLambda} is unique, the primal variable $b^*_n(\lambda(j))$ will also converge to the primal optimal variable $b^*_n$. 

Algorithm \ref{alg: distributed BA} summarizes the distributed inter-service bandwidth allocation (DISBA) algorithm. The algorithm works iteratively. In each iteration, the operator sends the current $\lambda(j)$ to all service providers. Then, each service providers solves for $b^*_n(\lambda(j))$ using its local information and sends the result to the network operator. The network operator finally updates $\lambda(j+1)$ for the next iteration's computation. The algorithm terminates until $\lambda$ converges. 

\begin{algorithm}
    \caption{Distributed Inter-Service Bandwidth Allocation (DISBA)} \label{alg: distributed BA}
    \begin{algorithmic}[1]
        \State \textbf{Input to Network Operator}: total bandwidth $B$, step size $\gamma$, convergence gap $\epsilon$
        \State \textbf{Input to service provider $n$}: FL service $n$ parameters $\langle s^\text{DT}_n, \{w^\text{LC}_{n,k}\}_{k\in\mathcal{K}_n}, s^\text{UT}_n, w^\text{GC}_n\rangle$, channel gains and computing speed of its clients $\mathcal{K}_n$. 
        \State \textbf{Initialization}: set $j = 0$ and $\lambda(0)$ equal to some non-negative value
        \While{$\lambda(j) - \lambda(j-1) > \epsilon$}
            \State Network Operator sends $\lambda(j)$ to all service providers
            \State Each service provider $n$ obtains $b^*_n(\lambda(j))$ by solving Eqn. \eqref{optBLambda} using bi-section
            \State Each service provider $n$ sends $b^*_n(\lambda(j))$ to Network Operator
            \State Network Operator updates $\lambda(j+1)$ according to Eqn. \eqref{updateLambda}
            \State $j \gets j + 1$
        \EndWhile
    \end{algorithmic}
\end{algorithm}

\section{Selfish service providers}
In the previous section, the distributed inter-service bandwidth allocation works by letting each FL service provider compute the allocated bandwidth $b^*_n(\lambda(j))$ given $\lambda(j)$. This, however, creates an opportunity for a selfish service provider to mis-report its computation result that favors itself but reduces the system performance as a whole. In fact, even if the inter-service bandwidth allocation problem \eqref{serviceBA} is solved in a centralized way, similar selfish behavior may still undermine the efficient system operation as a selfish service provider may mis-report its FL service and client parameters (e.g., FL workload, client computing power and channel gains etc.), which will alter the frequency function $f^*_n$ used at the operator side. With a wrong frequency function $f^*_n$, the operator will not be able to determine the \textit{true} optimal bandwidth allocation. 

In this section, we address the selfishness issue in inter-service bandwidth allocation by designing a multi-bid auction mechanism. This auction mechanism will ensure that the FL service providers are using their true FL frequency functions $f^*_n$ when making bandwidth bids. 

\subsection{Multi-bid Auction}
First, we describe the general rules of the  multi-bid auction mechanism. 

\subsubsection{Bidding}
At the beginning of each bandwidth allocation period, each service provider $n$ submits a set of $M$ bids $s_n = \{s^1_n, ..., s^M_n\}$. For each $m \in \{1, ..., M\}$, $s^m_n = (b^m_n, p^m_n)$ is a two-dimensional bid, where $b^m_n$ is the requested bandwidth and $p^m_n$ is the unit price that service provider $n$ is willing to pay to get the requested bandwidth $b^m_n$. Without loss of generality, we assume that bids are sorted according to the price such that $p^1_n \leq p^2_n \leq ... \leq p^M_n$. Let $S \in \mathbb{R}^+\times \mathbb{R}^+$ denote the set of multi-bids that a service provider can submit. 

\subsubsection{Bandwidth Allocation and Charges}
Once the network operator collects all multi-bids from all service providers, denoted by $s = \{s_n\}_{n \in \mathcal{N}}$, it computes and implements the inter-service bandwidth allocation $(b_1, ..., b_N)$. Each service provider $n$ then further allocates $b_n$ to its clients to perform FL. At the end of the period, the network operator determines the charges $(c_1, ..., c_N)$ for all service providers depending on the allocated bandwidth and the realized FL performance.

Now, a couple of issues remain to be addressed. First, how to compute the bandwidth allocation and determine the charges given the service provider-submitted multi-bids? Second, do the service providers have incentives to truthfully report their valuations of the bandwidth? These are the questions to be addressed in the next subsections.


\subsection{Market Clearing Prices with Full Information}
We first consider a simpler case where the service providers \textit{truthfully} report the \textit{complete} FL frequency function $f^*_n(b), \forall n$ to the network operator. This analysis will provide us with insights on how to design bandwidth allocation and charging rules in the more difficult multi-bid auction case. 

Recall that $f^*_n(b)$ is the optimal FL frequency of service $n$ if it has bandwidth $b$. Taking into account the price paid to obtain this bandwidth, the (net) utility of service provider $n$ is
\begin{align}
    u_n(b; p) = f^*_n(b) - p\cdot b
    \label{utility}
\end{align}
Now, if the bandwidth were sold at the unit price $p$, then service provider $n$ would buy $b_n(p) = \arg\max_b u_n(b; p)$ bandwidth in order to maximize its utility. We call $b_n(p)$ the \textbf{bandwidth demand function} (BDF), and it is easy to show that $b_n(p) = ({f^*_n}')^{-1}(p)$ by checking the first-order condition of Eqn. \eqref{utility}. On the other hand, if service provider $n$ requires a bandwidth $b$, then the service provider would pay a unit price no more than $p_n(b) = {f^*_n}'(b)$. We call $p_n(b)$ the \textbf{marginal valuation function} (MVF). 

\subsubsection{Market clearing price} With the complete information of $f^*_n(b)$ and hence BDF $b_n(p)$ for all service providers, the network operator can compute the \textbf{market clearing price} (MCP) $\rho$ so that $\sum_{n=1}^N b_n(\rho) = B$. One can prove that the MCP is unique and optimal in the sense that it maximizes the total (equivalently, average) FL frequency. 

\begin{proposition}
The market clearing price $\rho$ is unique and maximizes the total FL frequency $\sum_{n=1}^N f^*_n(b_n)$. 
\end{proposition}
\begin{proof}
According to Lemma 1, ${f^*_n}'(b)$ is an increasing function. Therefore, the BDF, which is the inverse function of ${f^*_n}'(b)$ is also increasing. As a result, there exists a unique solution to the increasing function $\sum_{n=1}^N b_n(p_n) = B$. 

To show that $\bar{p}$ maximizes $\sum_{n=1}^N f^*_n(b_n)$, consider the following maximization problem
\begin{align}
    \max_{b_1, ..., b_N}~~\sum_{n=1}^N f^*_n(b_n)~~~\text{subject to}~~\sum_{n=1}^N b_n = B
\end{align}
This is clearly a convex optimization problem. Consider its Karush-Kuhn-Tucker conditions. In particular, the stationarity condition is
\begin{align}
    \nabla \sum_{n=1}^N f^*_n(b_n) + \lambda \nabla (\sum_{n=1}^N b_n - B) = 0
\end{align}
where $\lambda$ is the Lagrangian multiplier associated with the constraint. The solution requires
\begin{align}
    {f^*_n}'(b_n) = \lambda, \forall n
\end{align}
Together with the feasibility constraint, this is equivalent to imposing a homogeneous market clearing price. \end{proof}

Because $b_n(p)$ is a monotonically decreasing function in $p$, a bi-section algorithm can be easily designed to find the unique market clearing price so that $\sum_{n=1}^N b_n(\rho) = B$. 

\subsubsection{Fairness-adjusted costs}
One major issue with the above pricing scheme is that it ignores fairness among the service providers: although it maximizes efficiency in terms of the average FL frequency according to Proposition 2, it is possible that the average FL frequency is maximized at an operating point where a few service providers are allocated with most of the bandwidth while some service providers obtain very little bandwidth. In this paper, we design and incorporate a fairness-adjusted charging scheme into the above pricing scheme. The payment of service provider $n$ now consists of two parts as follows:
\begin{itemize}
    \item The first part of the payment depends on the amount of bandwidth $b_n$ allocated to the service provider $n$, and the unit price $p$ set by the operator. Specifically, this payment is $p\cdot b_n$.
    \item The second part of the payment depends on the realized FL frequency $f_n$ of service provider $n$. Specifically, service provider $n$ will be charged a \textbf{fairness-adjusted cost} of $\alpha\cdot (f_n - \log (1 + f_n))$ at the end of the period once $f_n$ has been realized, where $\alpha \in [0, 1]$ is a tunable parameter. 
\end{itemize}

With these payments, service provider $n$'s utility becomes
\begin{align}
    u_n(b; p) = f^*_n(b) - p\cdot b - \alpha\cdot (f^*_n(b) - \log (1 + f^*_n(b))) = g_n(b) - p\cdot b
    \label{newutility}
\end{align}
where $g_n(b) \triangleq (1-\alpha) f^*_n(b) + \alpha \log (1 + f^*_n(b))$. Comparing this new utility function Eqn. \eqref{newutility} with Eqn. \eqref{utility}, we make the following remarks. First, the fairness-adjusted cost essentially replaces $f^*_n(b)$ with $g_n(b)$. The decision problem remains largely the same except that now we have a different benefit function. Second, in the new utility function Eqn. \eqref{utility}, given any allocated bandwidth $b$, it is still in the service provider's interest to perform the optimal client-level bandwidth allocation to maximize $f_n(b)$. This is because $g_n(b)$ is an increasing function in $f_n(b)$ for $\alpha \in [0,1]$. Therefore, we can directly write $g_n(b)$ as a function of the optimal FL frequency $f^*_n(b)$. Third, to charge the fairness-adjusted cost, the network operator does not need to know the exact function $f^*_n(b)$. Rather, it only has to know the realized FL frequency $f_n$ at the end of the current period. This is key to achieving fairness in multi-bid auction where FL service providers do not report the complete FL frequency function $f^*_n(b)$. 

We call $d_n(p) = (g'_n)^{-1}(p)$ the modified bandwidth demand function (mBDF). Likewise, we call $q_n(b) = g'_n(b)$ the modified marginal valuation function (mMVF). The network operator can similarly compute the modified market clearing price (mMCP) $\zeta$ so that $\sum_{n=1}^N d_n(\zeta) = B$. Using a similar argument that proves Proposition 2, one can prove Proposition 3 as follows.
\begin{proposition}
The mMCP $\zeta$ is unique and the resulting bandwidth allocation $(b_1, ..., b_N)$ maximizes $\sum_{n=1}^N \left[(1-\alpha)f^*_n(b_n) + \alpha \log (1 + f^*_n(b))\right]$. \end{proposition}
\begin{proof}
Because $f^*_n(b)$ is a concave increasing function, $\log (1 + f^*_n(b))$ is also concave and increasing. This further shows that $g_n(b)$ is concave and increasing. Following similar arguments in the proof of Theorem 2 proves the bandwidth allocation as a result of mMCP maximizes $\sum_{n=1}^N g_n(b_n)$. 
\end{proof}
The parameter $\alpha$ makes a tradeoff between efficiency and fairness. On the one hand, setting $\alpha = 0$ reduces the problem to the total FL frequency maximization problem. On the other hand, setting $\alpha = 1$ achieves proportional fairness among the service providers.

\subsection{Bandwidth Allocation and Charging Rules}
Now, we are ready to describe the bandwidth allocation and charging rules in fairness-adjusted multi-bid auction. In this subsection, each service provider $n$ submits only a multi-bid $s_n = (s^1_n, ..., s^M_n)$ instead of the complete FL frequency function $f^*_n(b)$. However, we will assume that the service providers are \textit{truthfully} submitting their bids, which will be proven indeed true in the next subsection. Specifically, we say that a bid $s^m_n = (b^m_n, p^m_n)$ is truthful if the bandwidth demand $b^m_n$ and the price $p^m_n$ that FL service provider $n$ is willing to pay satisfy the mBDF because it reveals FL service provider $n$'s true valuation of bandwidth after taking into consideration the fairness-adjusted costs. A multi-bid is truthful if all bids are truthful. 
\begin{definition}
(Truthful Multi-bid) A multi-bid $s_n = (s^1_n, ..., s^M_n)$ is truthful if $\forall m$, $s^m_n = (b^m_n, p^m_n)$ is such that $p^m_n = g_n'(b^m_n)$.
\end{definition}

The network operator does not know the BDF (and hence the mBDF) of each FL service provider $n$ because it does not have access to the FL frequency function $f^*_n$. Nonetheless, suppose service provider $n$ submitted a truthful multi-bid $s_n$, then the operator can compute a \textbf{pseudo-mBDF} using these bids to have some idea of the actual mBDF. Specifically, given the submitted multi-bid $s_n$, a left-continuous step function can be used to describe the pseudo-mBDF as follows,
\begin{equation}
\bar{d}_n(p) = \left\{
\begin{array}{ll}
       0, &\text{if}~p^M_n < p\\
      \max_{1 \leq m \leq M} \{b^m_n: p^m_n \geq p\}, &\text{otherwise}
\end{array}\right.
\end{equation}
Essentially, the pseudo-mBDF uses $b^m_n$ to approximate the bandwidth demand for prices in the range $(p^{m-1}_n, p^m_n]$. Similarly, the operator can also construct a \textbf{pseudo-mMVF} (pseudo-MVF), an approximation of service provider $n$'s actual mMVF using the submitted multi-bid, as follows,
\begin{equation}
\bar{q}_n(b) = \left\{
\begin{array}{ll}
       0, &\text{if}~b^1_n < b\\
      \max_{1 \leq m \leq M} \{p^m_n: b^m_n \geq b\}, &\text{otherwise}
\end{array}\right.
\end{equation}
In other words, the pseudo-mMVF uses $p^m_n$ to approximate the marginal value for bandwidth allocation in the range $[b^m_n, b^{m+1}_n)$. We illustrate the pseudo-mBDF and pseudo-mMVF in Figure \ref{fig:BDF_MVF}. 

\begin{figure}[htb]
	\centering
	\includegraphics[width=0.6\textwidth]{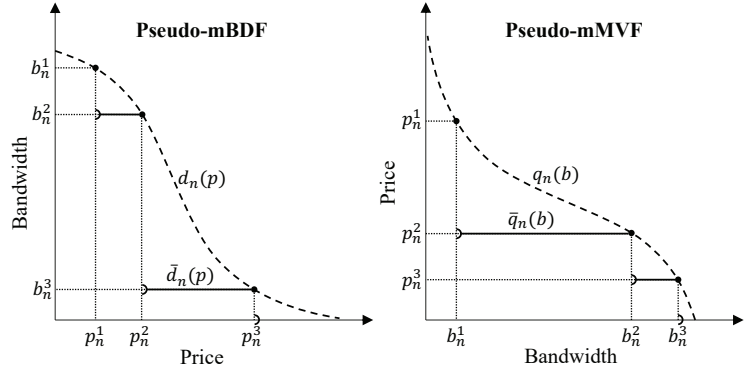}
	\caption{\small Pseudo-mBDF and Pseudo-mMVF.}
		\vspace{-0.1in}
	\label{fig:BDF_MVF}
\end{figure}

The \textbf{aggregated pseudo-mBDF} is the sum of pseudo-mBDFs of all FL service providers:
\begin{align}
\bar{d}(p) = \sum_{n=1}^N \bar{d}_n(p)
\end{align}
The \textbf{pseudo-mMCP} $\bar{\zeta}$ is the largest possible price so that the aggregated pseudo-mBDF exceeds the total available bandwidth, i.e.,
\begin{align}
    \bar{\zeta} = \sup\{p: \bar{d}(p) > B\}
\end{align}
This implies that reducing the mMCP by just a little bit will result in the supply (i.e., the total available bandwidth $B$) being no greater than the demand. Because every individual pseudo-mBDF function is a step function with $K$ steps, the aggregated pseudo-mBDF is also a step function with at most $NK$ steps. Therefore, the complexity of computing $\bar{\zeta}$ is at most $O(NK)$.

Next, we describe our bandwidth allocation and charging rules. For notational convenience, we denote $y(x^+) = \lim_{z \to x, z > x} y(x)$ when this limit exists for a function $y: \mathbb{R} \to \mathbb{R}$ and all $x \in \mathbb{R}$. 

\subsubsection{Bandwidth allocation}
With the pseudo-mMCP $\bar{\zeta}$, our bandwidth allocation rule is as follows: if FL service provider $n$ submits the multi-bid $s_n$ (and thereby declares the associated functions $\bar{d}_n$ and $\bar{q}_n$), then it receives bandwidth $b_n(s_n, s_{-n})$, with
\begin{align}
    b_n(s_n, s_{-n}) = \bar{d}_n(\bar{\zeta}^+) + \frac{\bar{d}_n(\bar{\zeta}) - \bar{d}_n(\bar{\zeta}^+)}{\bar{d}(\bar{\zeta}) - \bar{d}(\bar{\zeta}^+)}\left(B - \bar{d}(\bar{\zeta}^+)\right)
\end{align}
In other words: (1) Each FL service provider $n$ receives an amount of bandwidth it asks for at the lowest price $\bar{\zeta}^+$ for which supply exceeds the pseudo-bandwidth demand. (2) If all bandwidth is not allocated yet, the surplus $B - \bar{d}(\bar{\zeta}^+)$ is shared among service providers. This share is done proportionally to $\bar{d}_n(\bar{\zeta}) - \bar{d}_n(\bar{\zeta}^+)$ as we notice that $\bar{d}(\bar{\zeta}) - \bar{d}(\bar{\zeta}^+) = \sum_{n=1}^N \left(\bar{d}_n(\bar{\zeta}) - \bar{d}_n(\bar{\zeta}^+)\right)$, and ensures that all bandwidth is allocated. 

\subsubsection{Charging}
Given the submitted multi-bids $s$, each service provider $n$ is charged a payment $c_n(s)$ as follows,
\begin{align}
    c_n(s_n, s_{-n}) = \sum_{j\neq n}\int_{b_j(s)}^{b_j(s_{-n})} \bar{q}_j(b) db + \alpha\cdot (f^*_n(b_n) - \log (1 + f^*_n(b_n)))
\end{align}
The first term on the right-hand side is based on the \textit{exclusion-compensation} principle in second-price auction mechanisms \cite{vickrey1961counterspeculation}: service provider $n$ pays so as to cover the ``social opportunity cost'', namely the loss of utility it imposes on all other service providers by its presence. The second term on the right-hand side is the fairness-adjusted cost, which is charged at the end of each period after the actual FL frequency is realized and observed. 

Considering both the achieved FL frequency and the payment, FL service provider $n$'s utility is therefore
\begin{align}
u(s) = f^*_n(b_n(s)) - c_n(s)
\end{align}

\subsection{Incentives of Truthful Reporting}
In the previous subsection, we assumed that the every service provider truthfully submits its bid. Now, we prove that this assumption indeed ``approximately'' holds under the designed bandwidth allocation and charging rules.

We first study the individual rationality of the designed mechanism. 
\begin{definition}
A mechanism is said to be \textbf{individual rational} if no service provider can be worse off from participating in the auction than if it had declined to participate.  
\end{definition}

\begin{proposition}
If FL service provider $n$ submits a truthful multi-bid $s_n$, then $u_n(s) \geq 0$.
\end{proposition}
\begin{proof}
By Lemma 1, it is straightforward to see that $g_n(b)$ has the following properties:
\begin{itemize}
    \item $g_n(b)$ is differentiable and $g_n(0) = 0$
    \item $g'_n(b)$ is positive, non-increasing and continuous
    \item $\exists \gamma_n > 0$, $\forall b \geq 0$, $g'_n(b) = 0$ $\Rightarrow$ $\forall \tilde{b} < b$, $g'_n(b) \leq g'_n(\tilde{b}) - \gamma_n (b - \tilde{b})$. 
\end{itemize}
Therefore, $g_n(b)$ satisfies [Assumption 1, \cite{maille2004multibid}]. According to [Property 10, \cite{maille2004multibid}], we have
\begin{align}
    \sum_{j\neq n}\int_{b_j(s)}^{b_j(s_{-n})} \bar{q}_j(b) db \leq g_n(b_n(s_n, s_{-n}))
\end{align}
which is equivalent to $c_n(s_n, s_{-n}) \leq f^*_n(b_n(s_n, s_{-n}))$. Therefore, $u(s) \geq 0$. 
\end{proof}

Next, we show that truthful reporting is approximately \textbf{incentive compatible}, i.e., a service provider cannot do much better than simply reveal its true valuation. 
\begin{proposition}
Consider any truthful multi-bid $s_n$ for service provider $n$, and any other multi-bid $\tilde{s}_n \neq s_n$, $\forall s_{-n}$, we have
\begin{align}
    u_n(s_n, s_{-n}) \geq u_n(\tilde{s}_n, s_{-n}) - \Delta_n
\end{align}
where
\begin{align}
     \Delta_n = \max_{0\leq m \leq M}\int_{d_n(p^{m+1}_n)}^{d_n(p^m_n)}(q_n(b) - p^m_n)d b
\end{align}
with $p^{M+1}_n = q_n(0)$ and $p^0_n = p_0$. 
\end{proposition}
\begin{proof}
The proof follows [Proposition 2, \cite{maille2004multibid}].
\end{proof}

The above proposition shows that if service provider $n$ submits a truthful multi-bid $s_n$, then every other multi-bid $\tilde{s}_n$ necessarily corresponds to an increase of utility no larger than $\Delta_n$. In other words, a truthful bidding brings service provider $n$ the best utility possible up to a gap $\Delta_n$. Importantly, this value does not depend on the number of other service providers or the multi-bids they submit. In the game theoretic terminology, the situation where all service providers submit truthful multi-bids is an \textit{ex post $\Delta$-Nash equilibrium}, where $\Delta = \max_n \Delta_n$, in the sense that no service provider could have improved its utility by more than $\Delta$ if it had submitted a different multi-bid. 

\subsection{An Uniform Multi-Bidding Example}
To conclude the multi-bid auction mechanism design, we illustrate a uniform multi-bidding approach as an example of how to decide the multi-bid of an individual service provider. Instead of having the service provider submitting both prices and bandwidth requests, the operator can announce $M$ prices $(p^1_n, ..., p^M_n)$ to service provider $n$ and let service provider $n$ report its requested bandwidth $(b^1_n, ..., b^M_M)$ at these price points. This way, the operator has a better control over how the service providers make multi-bids to avoid multi-bids that may result in a large $\Delta_n$, which may reduce service provider's incentives to truthfully report. Because the operator does not know the demand function of service provider $n$, a natural approach is to uniformly distribute these $M$ prices in the range $[p_0, p^\text{max}_n]$ where $p^\text{max}_n$ is the largest price at which the service provider may still request a positive amount of bandwidth. Specifically, 
\begin{align}
    p^\text{max}_n = p_n(0) = {f^*_n}'(0) = \left(\sum_{k=1}^{K_n} \alpha_{n,k} \right)^{-1} = \left(\sum_{k=1}^{K_n} (\frac{s^\text{DT}_n}{r^\text{DT}_k} + \frac{s^\text{UT}_n}{r^\text{UT}_k}) \right)^{-1} 
\end{align}
Assume that the network operator has prior knowledge $\underline{K}_n$, $\underline{s}^\text{DT}_n$, $\underline{s}^\text{UT}_n$, $\bar{r}^\text{DT}_n$ and $\bar{r}^\text{UT}_n$ on the lower/upper bounds on the parameters, then $p^\text{max}_n$ can be upper bounded by
\begin{align}
p^\text{max}_n \leq \underline{K}_n^{-1}\cdot \left( \frac{\underline{s}^\text{DT}_n}{\bar{r}^\text{DT}_n} + \frac{\underline{s}^\text{UT}_n}{\bar{r}^\text{UT}_n}\right)^{-1} \triangleq \bar{p}^\text{max}_n
\end{align}
Thus, the operator can set the uniform prices as
\begin{align}
    p^m_n = p_0 + m\cdot \frac{\bar{p}^\text{max}_n - p_0}{M + 1}, \forall m = \{1, ..., M\}
\end{align}
Note that there is an intrinsic trade-off on the choice of $M$. On the one hand, a large $M$ allows the pseudo-BDF and pseudo-MVF to more accurately reflect the true BDF and MVF at an increased complexity and signaling overhead. On the other hand, a smaller $M$ makes multi-biding easier but the discrepancy between the pseudo functions and the true functions will introduce a larger performance loss.

\section{Simulations}
In this section, we conduct simulations to evaluate the performance of the proposed methods. 

\subsection{Simulation Setup}
The simulated wireless network adopts an OFDMA system with a total bandwidth of $B = 10$ MHz. The period length is set as $T = 20 s$. The number of clients of a FL service is drawn from a Normal distribution with mean 25. In every period, a new FL task may start following a scheduled plan, which is defined by a Poisson distribution with the mean interval $p_\text{arrive}$. By tuning $p_\text{arrive}$, we adjust the FL service demand, and a smaller $p_\text{arrive}$ will more likely lead to more concurrent FL services in a period as an FL service often lasts multiple periods. Each FL service has a pre-determined target training accuracy, and when the accuracy reaches the target, the FL service terminates and exits the wireless network. The clients' wireless channel gain is modeled as independent free-space fading where the average path loss is from a Normal distribution with different mean and variance in different circumstances. The variance of the complex white Gaussian channel noise is set as $10^{-12}$. For each client, the local training time is uniformly randomly drawn from  $[0.01, 0.05]$ s. We fix the global aggregation time to be $1 \times 10^{-5}$. We consider typical neural network sizes in the range of $[0.2, 0.5]$ Mbits. The upload transmission power is uniformly randomly between 0.05 and 0.15 W, and the download transmission power is uniformly randomly between 0.1 and 0.3 W.

\subsection{Convergence of DISBA in the Cooperative Case}
We first illustrate the convergence behavior of DISBA in the cooperative FL service provider case in a representative period with 5 concurrent FL services. These services have 10, 12, 14, 16, 18 clients, respectively. In Figure \ref{fc}, we show the computed FL frequency for each service provider before convergence. As Figure \ref{bc} shows, the bandwidth allocation quickly converges to the optimal allocation for a convergence tolerance gap $\epsilon = 1e-3$. Eventually, the resulting FL frequencies of these FL services in this period are reported in Table \ref{table0}.  We further show in Table \ref{table1} the computation time of DISBA for different values of the tolerance gap and step size. The time values are measured on a desktop computer with Intel Core i5-9400 2.9GHz GPU and 16GB memory.

\begin{figure}[htbp]
\centering
\begin{minipage}[t]{0.48\textwidth}
\centering
\includegraphics[width=7.5cm]{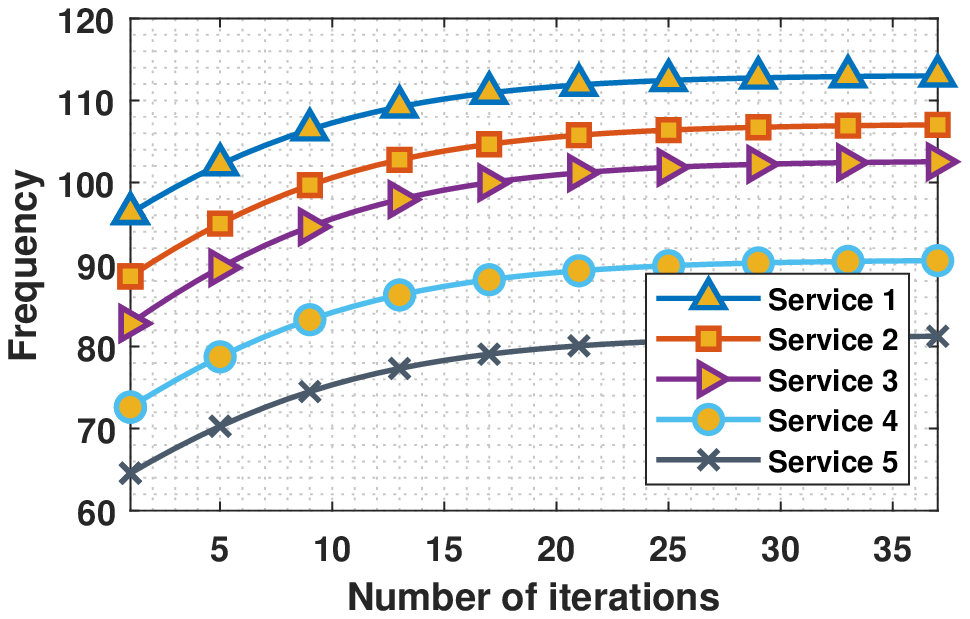}
\setlength{\abovecaptionskip}{0pt}
\caption{\label{fc} Frequency of Each Service before Convergence}
\end{minipage}
\begin{minipage}[t]{0.48\textwidth}
\centering
\includegraphics[width=7.5cm]{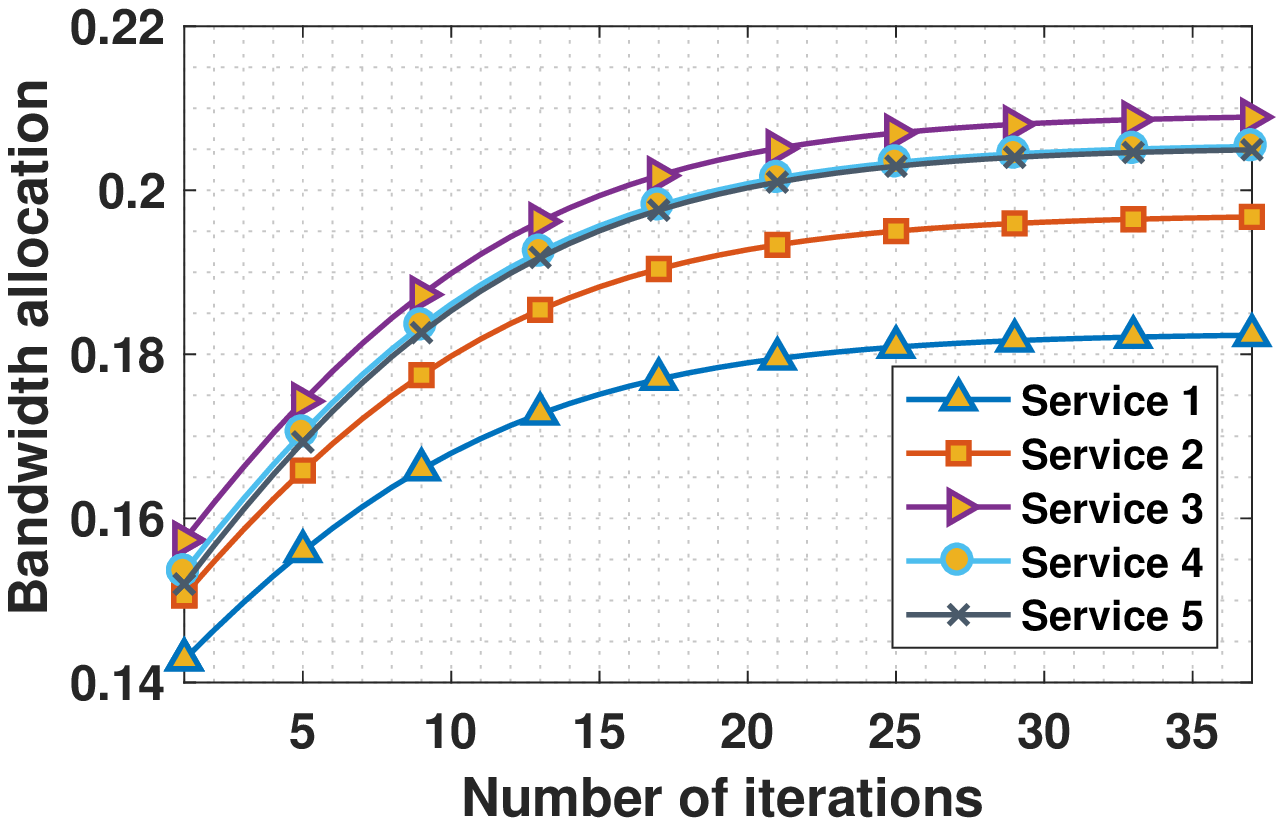}
\setlength{\abovecaptionskip}{0pt}
\caption{\label{bc}Bandwidth of Each Service before Convergence}
\end{minipage}
\end{figure}

\begin{table}[ht]
\centering
\begin{tabular}{|c c c c |} 
 \hline
 Service Index & Number of Clients & Bandwidth Ratio & Frequency \\ 
 \hline
1 & 10 & 0.182 & 113\\ 
2 & 12 &  0.196 & 107 \\
3 & 14 &  0.209 &  102.6  \\
4 & 16 &  0.205 & 90.4  \\
5 & 18 & 0.205 & 81.2 \\
 \hline
\end{tabular}
\vspace{0.1 in}
\caption{Resulted Bandwidth Allocation and Frequency of Each FL Service (Cooperative)}
\label{table0}
\end{table}

\begin{table}[ht]
\centering
\begin{tabular}{|c c c c|} 
 \hline
 Tolerated Gap & Step Size & \# of Iterations & Time(s) \\ 
 \hline
 1e-3 & 0.1 & 131 & 0.332 \\ 
 1e-3 & 0.5 & 37 & 0.094 \\
 5e-3 & 0.1 & 72 & 0.169 \\
 5e-3 & 0.5 & 26 & 0.069 \\
 \hline
\end{tabular}
\vspace{0.1 in}
\caption{Computational Complexity for the Cooperative Provider Case}
\label{table1}
\end{table}

\subsection{Fairness-adjusted Multi-bid Auction in the Selfish Case}
We perform fairness-adjusted multi-bid auction in the same representative period as in the last subsection, with $M = 5$ and $\alpha = 0.5$. The pseudo-mBDFs of the FL service providers and the aggregated pseudo-mBDF are illustrated in Figures \ref{pd} and \ref{pc}, respectively. The pseudo-MCP is also shown in Figure \ref{pc}. Table \ref{table1.5} reports the resulting bandwidth allocation and achieved FL frequency. 

\begin{figure}[htbp]
\centering
\begin{minipage}[t]{0.48\textwidth}
\centering
\includegraphics[width=7.5cm]{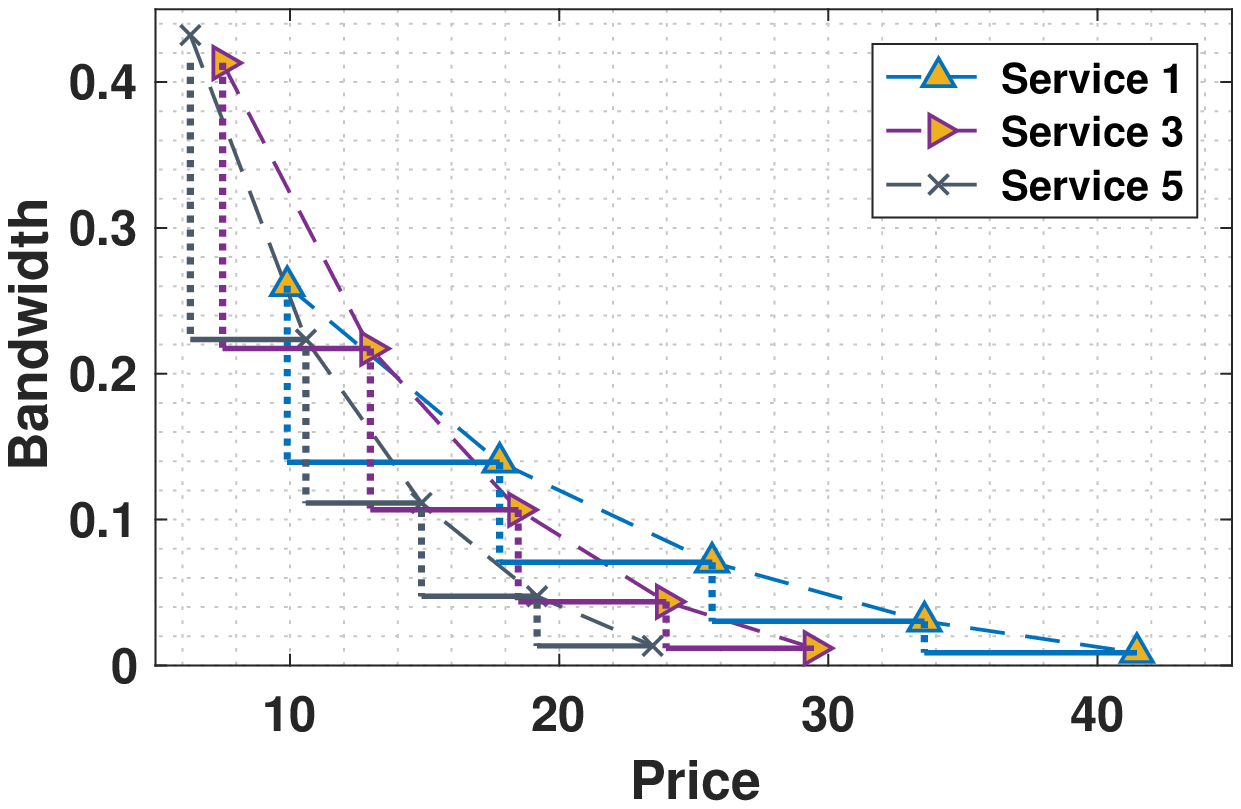}
\setlength{\abovecaptionskip}{0pt}
\caption{\label{pd} Pseudo-mBDF of Individual FL Services}
\end{minipage}
\begin{minipage}[t]{0.48\textwidth}
\centering
\includegraphics[width=7.5cm]{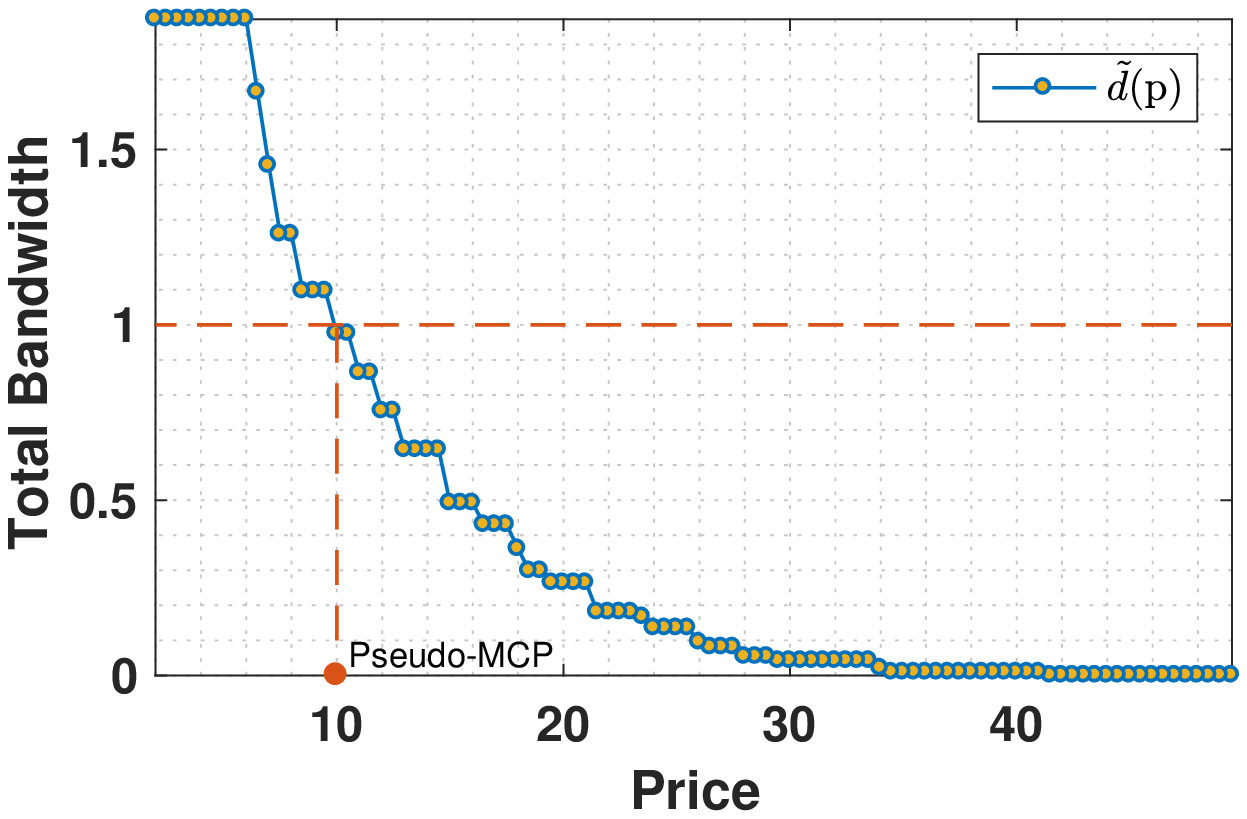}
\setlength{\abovecaptionskip}{0pt}
\caption{\label{pc}Aggregated Pseudo-mBDF and Pseudo-MCP}
\end{minipage}
\end{figure}

\begin{table}[ht]
\centering
\begin{tabular}{|c c c c |} 
 \hline
 Service Index & Number of Clients & Bandwidth Ratio & Frequency \\ 
 \hline
1 & 10 & 0.164 & 105.82 \\ 
2 & 12 & 0.177 &  99.52 \\
3 & 14 &  0.217 &  105.46 \\
4 & 16 &   0.218 &   94.4 \\
5 & 18 & 0.223  &  86.56 \\
 \hline
\end{tabular}
\vspace{0.1 in}
\caption{Optimal Bandwidth and Frequency of Each Service (Selfish)}
\label{table1.5}
\end{table}

As we briefly mentioned in Section V, there is a trade-off when selecting the number of bids $M$. On the one hand, a larger $M$ increases the computational complexity for searching for the pseudo-MCP and determining the eventual bandwidth allocation. On the other hand, a larger $M$ improves the precision of the pseudo-MCP, thereby improving the allocation performance. In Figure \ref{error}, we demonstrate the overall performance by varying $M$. As can be seen, as $M$ increases, the overall performance will increase while each FL service provider needs to submit more bids to the server which will cause transmission delays and data backlogs.

\begin{figure}[htbp]
\centering
\includegraphics[width=7.5cm]{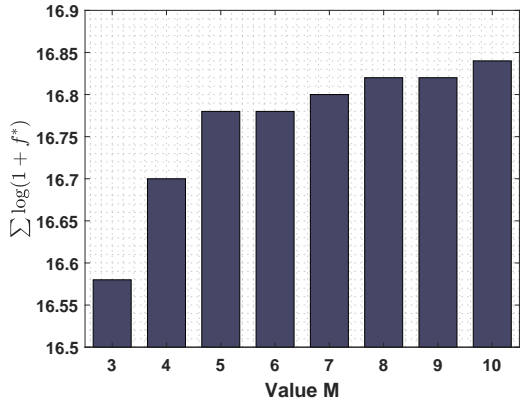}
\caption{\label{error}Overall Performance in the Selfish Service Providers Case with Different $M$}
\end{figure}

The parameter $\alpha$ plays an important role in the selfish owner case, which makes a tradeoff between efficiency and fairness. With a larger $\alpha$, the whole system sees fairness as more important, and conversely, the whole system is more concerned with the overall efficiency. The market clearing price is reflected in Figure \ref{cpa} and the overall utility is shown in Figure \ref{tua}. With the increase of $\alpha$, the market clearing price and the total utility will decrease, which can be treated as a concession to achieve fairness between different FL services.

\begin{figure}[htbp]
\centering
\begin{minipage}[t]{0.48\textwidth}
\centering
\includegraphics[width=7.5cm]{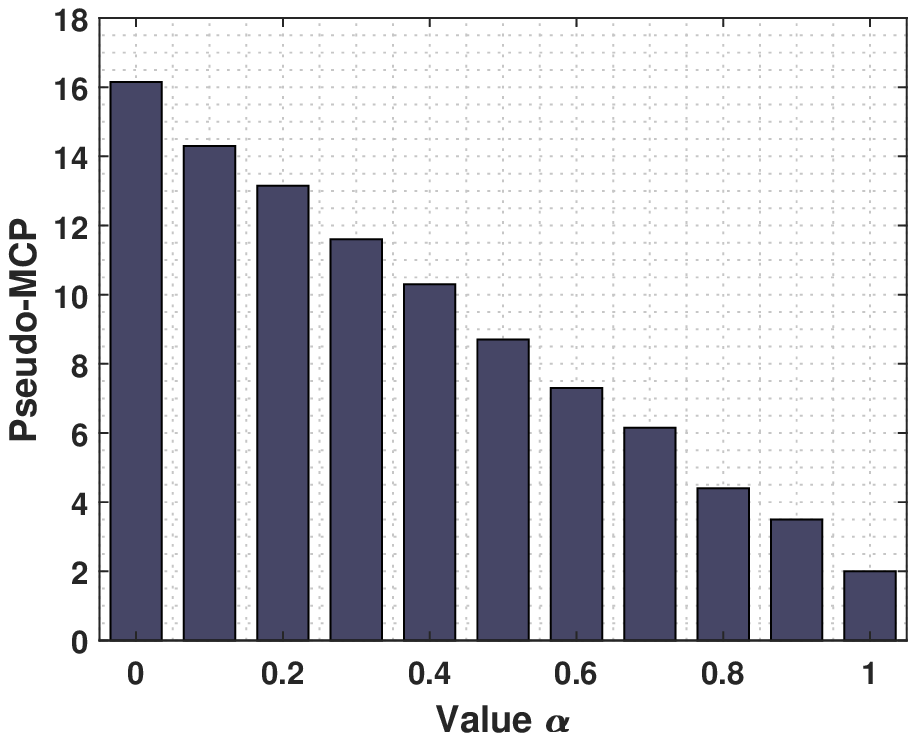}
\setlength{\abovecaptionskip}{0pt}
\caption{\label{cpa}  Pseudo-MCP with Different $\alpha$}
\end{minipage}
\begin{minipage}[t]{0.48\textwidth}
\centering
\includegraphics[width=7.5cm]{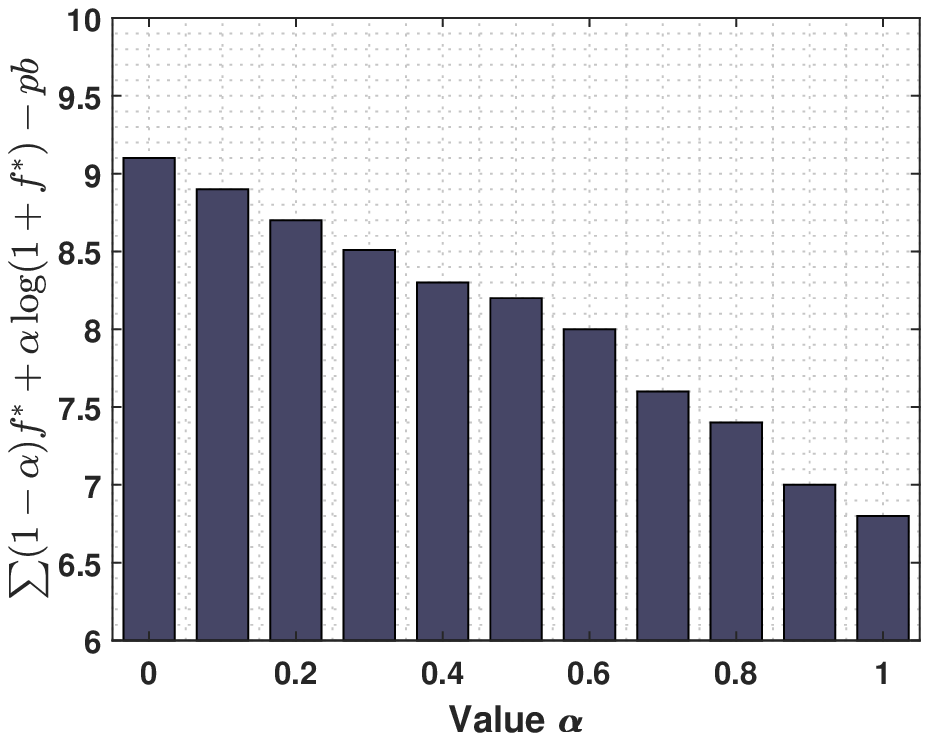}
\setlength{\abovecaptionskip}{0pt}
\caption{\label{tua} Total Utility with Different $\alpha$}
\end{minipage}
\end{figure}

\subsection{Performance Comparison}
In the following experiments, we compare our proposed algorithms with three benchmark algorithms. 
\begin{itemize}
    \item \textbf{Equal-Client (EC)}: Bandwidth is equally allocated to the clients. Therefore, each client gets a bandwidth of $B/\sum_n K_n$.
    \item \textbf{Equal-Service (ES)}: Bandwidth is equally allocated to the FL services. That is, each FL service gets a bandwidth of $B/N$. However, each FL service provider still performs the optimal intra-service bandwidth allocation among its clients.
    \item \textbf{Proportional (PP)}: Each FL service obtains a bandwidth that is proportional to the number of its client. That is, FL service $n$ obtains a bandwidth of $\frac{K_n}{\sum_j K_n} B$. This bandwidth is further allocated among its clients following the optimal intra-service bandwidth allocation. 
\end{itemize}

We start by comparing the proposed algorithms with benchmarks in the per-period setting. The overall performance is shown in Figure \ref{per-case}. In this setting, there are five FL services with a random number of clients drawn from a Normal distribution with mean 20 and variance 10 and random channel conditions drawn from a  Normal distribution with mean  85  and variance  15, and the result is averaged over 20 runs. As can be seen, our DISBA algorithm for the cooperative case (labeled as Coop) has the best performance, and the auction mechanism for the selfish case (labeled as Self) also outperforms the other benchmarks. Although ES and PP also perform the intra-service bandwidth allocation, the heterogeneity of the client number and channel conditions render them suboptimal. 

\begin{figure}[htbp]
\centering
\includegraphics[width=8cm]{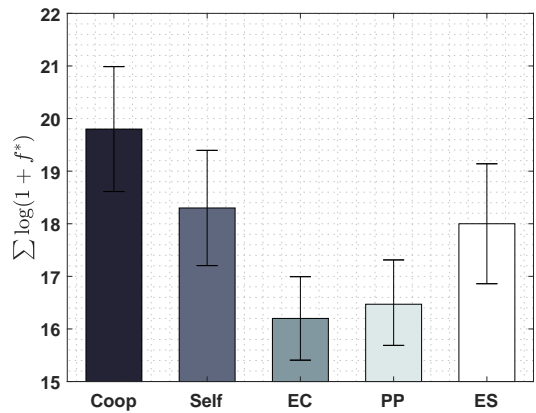}
\caption{\label{per-case} Per-period FL Performance of Different Algorithms}
\end{figure}

Because FL is a long-term process, we further investigate the long-term performance of the proposed algorithms. In the long-term setting, 10 FL services join the wireless network at different times controlled by the $p_\text{arrive}$-parameterized Poisson process and the FL service will be removed from the wireless network when its test accuracy has converged. Although the convergence of FL is complexly affected by many factors including the adopted FL algorithm, dataset and the selected clients, we assume that each of these 10 FL services require 2000 FL rounds, which is a typical value observed in the literature \cite{konevcny2016federated}, to reach convergence in order to provide a meaningful comparison of the algorithms in a controlled environment. Whenever a FL service has been run for 2000 rounds, it exits the system. 

Figure \ref{long-case1} illustrates the average duration (in terms of the number of periods) of all FL services by running different algorithms for $p_\text{arrive} = 5$, where the client number of a FL service is drawn from a Normal distribution with mean 25 and variance 15 and the channel condition of a FL service is drawn from a  Normal  distribution  with mean 85 and variance 15. The results are averaged over 20 runs. We can see that the proposed algorithms achieve the smallest average duration compared to the benchmarks, confirming their fast FL convergence even in the long-run. 

\begin{figure}[htbp]
\centering\includegraphics[width=8cm]{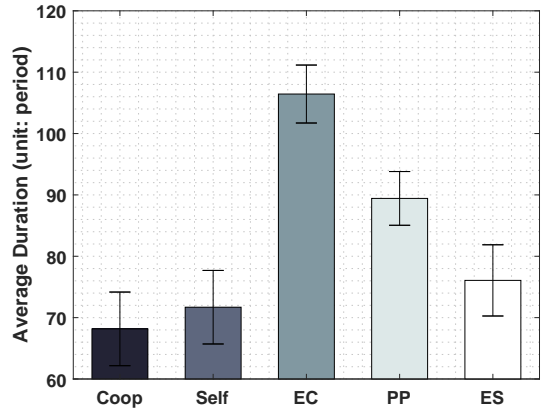}
\caption{\label{long-case1}Average Duration Period of FL Services}
\end{figure}

Next, we study the impact of the client number heterogeneity (which reflects the FL service size heterogeneity) on the performance of different algorithms. To this end, the client number of a FL service is drawn from a Normal distribution with mean 25 and we change the variance between 0 and 15 to adjust the heterogeneity degree. The result is shown in Figure \ref{long-client}: as the variance increases (i.e. a higher degree of heterogeneity), the mean of the average duration decreases, while the standard deviation of average duration increases. This is understandable because a higher degree of heterogeneity causes wireless bandwidth to be more unevenly distributed among the FL services, thereby degrading the overall FL performance. Notably, the performance gain of our proposed algorithms increases as the variance increases, which demonstrates the superior ability of our algorithms to handle the heterogeneous case.

\begin{figure}[htbp]
\centering\includegraphics[width=8cm]{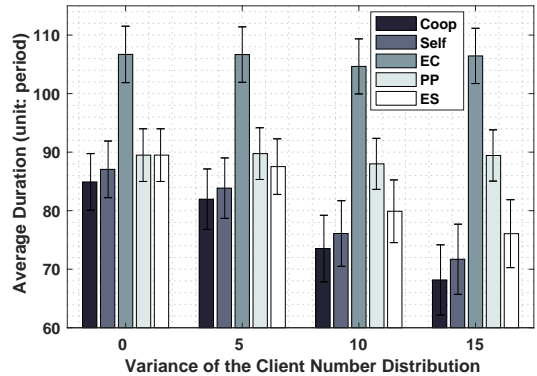}
\caption{\label{long-client}Impact of Client Number Heterogeneity}
\end{figure}

Furthermore, we also investigate the impact of the channel condition heterogeneity on the FL performance. In these simulations, the average channel condition of a FL service is drawn from a Normal distribution with mean 85 and we change the variance between 0 and 15 to adjust the heterogeneity degree. The channel conditions of clients of this FL are further drawn from a Normal distribution with a mean being the instantiated average channel condition. In Figure \ref{long-channel}, we observe a similar phenomenon as in Figure \ref{long-client}, which further confirms the advantage of adopting our proposed algorithms.

\begin{figure}[htbp]
\centering\includegraphics[width=8cm]{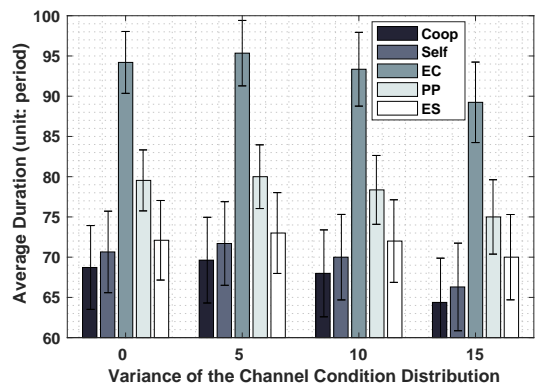}
\caption{\label{long-channel}Impact of the Channel Condition Heterogeneity}
\end{figure}

Finally, we study the influence of the mean arrival interval parameter  $p_\text{arrive}$ on the resulting average FL duration. in Figure \ref{long-case1_avg}, with the increasing of $p_\text{arrive}$,  the average duration of the FL services decreases. This is because when $p_\text{arrive}$ is small, many FL services pile up and co-exist in the wireless network, thereby reducing the wireless bandwidth an individual FL service can receive. 


\begin{figure}[htbp]
\centering
\includegraphics[width=7cm]{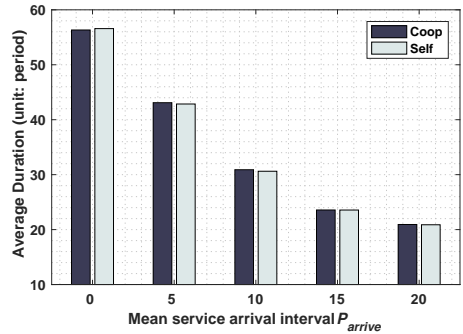}
\setlength{\abovecaptionskip}{0pt}
\caption{\label{long-case1_avg} Average Duration With Varying $P_{arrive}$}
\end{figure}

\section{Conclusion}
This paper studied a bandwidth allocation problem for multiple FL services in a wireless network, which has not been well studied in the literature. The considered problem consists of two interconnected subproblems, intra-service resource allocation, and inter-service resource allocation. By solving these problems, we optimally allocate bandwidth resources to multiple FL services and their corresponding clients to speed up the training process and meanwhile guarantee fairness for both cooperative and selfish FL service providers cases. Our method has shown superior performance compared to the benchmarks. However, there are several future research works that can be done to extend the impact of this work. For example, this paper takes FL frequency as the key metric to be optimized, but the true performance of FL is affected by the dataset, federated optimization algorithm, and many others. In addition, when a client can simultaneously participate in multiple FL services, resource allocation has to consider both the wireless bandwidth and client computing resources. 
\ifCLASSOPTIONcaptionsoff
  \newpage
\fi

\bibliographystyle{IEEEtran}
\bibliography{bibligraphy}

%

%




\end{document}